\documentclass[10pt]{article}

\usepackage{amssymb}
\usepackage{amsfonts}
\usepackage{amsmath}
\usepackage{amscd}
\usepackage[all,cmtip]{xy}
\usepackage{amsthm}
\usepackage{setspace}
\usepackage{enumerate}
\usepackage{geometry}
\usepackage{tcolorbox}
\usepackage{hyperref}
\usepackage{tikz}
\usepackage{pgfplots}

\hypersetup{
allbordercolors=blue, pdfborder=0 0 1.3}

\theoremstyle{plain}
\newtheorem{thm}{Theorem}
\newtheorem{lem}{Lemma}

\newtheorem{cor}[subsection]{Corollary}

\theoremstyle{definition}
\newtheorem*{example}{Example}
\newtheorem{question}{Question}

\title{Weyl orbit particles}
\author{Martin T.  Luu$^1$ \footnote{matluu@ucdavis.edu}} 

\date{\small{$^1$ Department of Mathematics, University of California, Davis,  CA 95616, United States of America}}

\begin{document}

\maketitle

\begin{abstract}
The mass spectrum of affine Toda theory is known to be expressible in terms of a suitable eigenvector of the relevant Cartan matrix.  The particles correspond in a precise manner to the Coxeter element orbits in the set of roots. Recently, variants of affine Toda theory have been constructed for many different Weyl group elements.  Again,  the particles correspond to orbits in the set of roots and this allows the calculation of the classical mass spectrum.  We show how these spectral calculations generalize the affine Toda relation with the Cartan matrix.  As an example, we calculate the spectrum for the unique non-Coxeter infinite family $\textrm{D}_{2n}(a_{n-1})$ of primitive regular conjugacy classes in the Weyl groups of complex simple Lie algebras.
\end{abstract}

\section{Introduction}

Let $\mathfrak g$ be a simple finite-dimensional complex Lie algebra and let $\sigma_{C}$ be a Coxeter element in the Weyl group of $\mathfrak g$.  Fring-Liao-Olive showed in \cite{FLO} that the particles of the affine Toda theory associated to $\mathfrak g$ correspond to the orbits of the cyclic group $\langle \sigma_{C} \rangle$ acting on the set of roots.  A beautiful consequence is an expression of the mass spectrum in terms of the entries of a suitable eigenvector of the Cartan matrix of $\mathfrak g$.  See the important experimental work by Coldea et. al.  \cite{COL} for an observation of parts of this spectrum in quasi $1$-dimensional magnetic systems related to $\mathfrak e_8$ affine Toda theory.  A mathematical survey of these developments is given in \cite{BG}.

Freeman introduced in \cite{FRE} a very useful re-formulation of the affine Toda theory Lagrangian. This involves Kostant's notion of a cyclic element in $\mathfrak g$, introduced in \cite{KOS}. Fix a Cartan algebra $\mathfrak h$ in $\mathfrak g$, and for $1 \le i \le r=\textrm{rank } \mathfrak g$ let $e_i$ be generators of the root spaces with respect to a choice of simple roots.  Further, let $e_0$ be a generator of the lowest root space.  The corresponding cyclic element is defined as 
$$\Lambda_{+}=\sum_{i=0}^{r}e_i$$ 
Kostant showed in \cite{KOS} that its centralizer is again a Cartan algebra $\mathfrak h'$,  different from $\mathfrak h$. Consider a field $\phi : \mathbb{R}^{2} \rightarrow \mathfrak h$ and let $(-,-)$ denote the Killing form.  Freeman showed that for a suitable element $\Lambda_{-}$ in $\mathfrak h'$,  affine Toda theory can be formulated in terms of the Lagrangian
$$ \frac{1}{2}(\partial_{\mu} \phi, \partial^{\mu} \phi)  - (\exp(\textrm{ad } \phi )(\Lambda_{+}) , \Lambda_{-})$$
In this description,  the Weyl group assumes a bigger role: both $\Lambda_{+}$ and $\Lambda_{-}$ are eigenvectors of a Coxeter element.  In \cite{LUU} we used this formulation to introduce new Lagrangians based on eigenvectors of Weyl group elements $\sigma$ different from Coxeter elements. We call these Toda-Weyl theories and the classical mass spectrum is calculated in \cite{LUU}. Again, the particles of the theory correspond to orbits (in the set of roots) of the cyclic group generated by $\sigma$.  These Weyl orbit particles have a mass with the following Lie theoretic expression: If $\gamma_1,\gamma_2,\cdots$ are orbit representatives, then by \cite{LUU} (Theorem 1) the masses are given by
$$|\Lambda_{+} \cdot \gamma_i|$$
The task is then to find an explicit description of these values, for example mimicking the description in terms of eigenvectors of the Cartan matrix if $\sigma$ is a Coxeter element.  In \cite{LUU} we gave some examples of this, without a systematic description of the eigenvectors of Weyl group elements in terms of eigenvectors of analogues of the Cartan matrix.  In the present work we give the full details of this correspondence.  In Section \ref{Carter-Weyl-section} the mathematical background is described,  and in Section \ref{example-section} this is used to give the full mass spectrum for a particularly interesting infinite family of Toda-Weyl theories associated to simple Lie algebras of type D: We consider a Weyl group element $\sigma$ in the conjugacy class $\textrm{D}_{2n}(a_{n-1})$ (see Section \ref{example-section} for the definition).  More specifically, we assume $\Lambda_{+}$ is an eigenvector of $\sigma$ with eigenvalue $e^{2 \pi i/2n}$.  The eigenspace is $2$-dimensional, and for each choice of eigenvector one obtains a Toda-Weyl theory.

For simple Lie algebras of type $\textrm{D}_{2n}$ the Coxeter element is of order $4n-2$ and the usual affine Toda mass spectrum has an elegant trigonometric expression
$$\sin\left (\frac{\pi}{4n-2} \right ), \cdots, \sin\left (\frac{(2n-2)\pi}{4n-2} \right ), \frac{1}{2} \times \bf{2}$$
where we adopt the notation $m \times \textbf{i}$ to signify that the value $m$ occurs with multiplicity $\textbf{i}$.  See for example \cite{DOR} (Table 2) for the relevant eigenvector calculations.  We show that for the conjugacy class $\textrm{D}_{2n}(a_{n-1})$ analogous trigonometric formulas exist as well.

\begin{thm}
\label{type-D-theorem}
Let $\sigma$ be a Weyl group element in the conjugacy class $\textrm{\emph{D}}_{2n}(a_{n-1})$.  There are basis vectors $\textrm{\emph{D}}_{2n}(a_{n-1})_{\textrm{\emph{\textbf{I}}}}$ and $\textrm{\emph{D}}_{2n}(a_{n-1})_{\textrm{\emph{\textbf{II}}}}$ of the $2$-dimensional $e^{2 \pi i/2n}$-eigenspace of $\sigma$,  such that the corresponding Toda-Weyl theories have the following mass spectrum:

\renewcommand*{\arraystretch}{2.2}

$$\begin{array}{|c||c|c|c|c|c|c}
\hline
 \textrm{\emph{D}}_{2n}(\textrm{\emph{a}}_{n-1})_{\textbf{\emph{I}}} & 0 & \sin \left (\frac{\pi}{2n} \right  ) \times \bf{4}  & \cdots & \sin \left (\frac{(n-1)\pi}{2n} \right ) \times \bf{4} & 1 \\
 \hline 
\textrm{\emph{D}}_{2n}(\textrm{\emph{a}}_{n-1})_{\textbf{\emph{II}}} & 1  \times  \bf{(4n-2)}\\  \cline{1-2}
\end{array}$$
\end{thm}

\renewcommand*{\arraystretch}{1.0}

\hspace{0.2in}

\section{Carter-Weyl correspondence}
\label{Carter-Weyl-section}

In this section we describe the mathematical background that allows us to relate the Toda-Weyl masses to eigenvectors of generalized Cartan matrices.

Let $\sigma$ be an element of the Weyl group $W$ of a finite-dimensional  simple complex Lie algebra $\mathfrak g$.  We are interested in a subtle relation between the $\sigma$ action on the root space, and the $\sigma$ action on the set of roots.  More precisely,  let $\Phi$ be the set of roots,  and decompose it into orbits
$$\Phi= \mathcal O_{1} \sqcup \cdots \sqcup \mathcal O_{s}$$
under the action of the cyclic group generated by $\sigma$. Let $\gamma_1,\cdots,\gamma_s$ be orbit representatives.  View $W$ as acting on a Cartan algebra $\mathfrak h$ and view the roots as elements of $\mathfrak h$ via the Killing form.  Now let $\Lambda$ in $\mathfrak h$ be an eigenvector of $\sigma$.

\begin{question}
\label{main-question}
What are the values $|\Lambda \cdot \gamma_{i}|$ for $1\le i \le s$?
\end{question}
We call these the Weyl-Orbit values.  Since the Killing form is invariant under Weyl group elements,  these absolute values are independent of the choice of orbit representatives.  When the Weyl group element $\sigma$ is a Coxeter element,  the Weyl-Orbit values are very interesting.  The number of orbits $s$ equals the rank of $\mathfrak g$,  and it is known that the Weyl-Orbit values correspond to the entries of an eigenvector of the Cartan matrix of $\mathfrak g$.  This was used by physicists to calculate the affine Toda theory mass spectrum, see \cite{DOR}, \cite{FRE}, \cite{FLO}.  Put differently, the relative geometry of the simple roots yields the answer to Question \ref{main-question} for Coxeter elements.  In the present work we generalize this phenomenon beyond the Coxeter case, by showing in what way the Weyl-Orbit values are governed by the relative geometry of certain special sets of roots.

The main idea is to generalize work of Coxeter \cite{COX} that relates the linear algebra of Cartan matrices and Coxeter elements,  a generalised correspondence we call the Carter-Weyl correspondence.  As noted in the work of Berman-Lee-Moody \cite{BLM}, this circle of ideas could be described as mathematical folklore,  and many aspects are described in detail in \cite{BLM}. However, the results are not quite in the form most convenient for the present purposes and we therefore present the Carter-Weyl correspondence from scratch, without any claim to originality. 

For a root $\alpha$ denote by $r_{\alpha}$ the corresponding reflection in the hyperplane orthogonal to $\alpha$.  For a Weyl group element $\sigma$ choose roots $\phi_{1}, \cdots, \phi_{k}$ such that
$$\sigma = r_{\phi_{1}} \cdots r_{\phi_{k}}$$
In \cite{CAR}, Carter defines $l(\sigma)$ to be the smallest possible such $k$ and  shows that 
$$l(\sigma) \le \textrm{rank }\mathfrak g$$
Suppose this inequality is strict. Then as in \cite{CAR} (Lemma 9), there is a rank $l(\sigma)$ subroot system of the root system of $\mathfrak g$, such that $\sigma$ is in the corresponding Weyl subgroup. Hence, in the following we  restrict our attention to Weyl group elements such that $l(\sigma) = \textrm{rank }\mathfrak g$.  It then follows from the work of Carter \cite{CAR} (Section 3), that there are mutually orthogonal roots $\xi^{+}_{1},\cdots, \xi^{+}_{a}$ and mutually orthogonal roots $\xi^{-}_{1},\cdots, \xi^{-}_{b}$ such that $a+b = \textrm{rank } \mathfrak g$ and $\sigma= \sigma_{+} \sigma_{-}$ where
\begin{eqnarray}
\label{sigma-decomposition}
\sigma_{\pm}= \prod_{i} r_{\xi^{\pm}_{i}}
\end{eqnarray}
We call such a collection of roots a set of Carter roots for $\sigma$.  Note that since $l(\sigma)=\textrm{rank }\mathfrak g$,  it follows from \cite{CAR} (Lemma 3) that the Carter roots are linearly independent.  Order them arbitrarily as $\xi_1,\xi_2,\cdots, \xi_{r}$ and define $c_i =\pm 1$ depending on whether $\xi_i$ is of the form $\xi_{j}^{+}$ or $\xi_{j}^{-}$.  Further,  define the Carter matrix $K$ in complete analogy with the Cartan matrix via 
$$K_{i,j} = 2  \; \xi_i \cdot \xi_j  /\xi_j \cdot \xi_j$$
We can now state the generalization of the relation between Coxeter elements and Cartan matrices.
\begin{thm}[Carter-Weyl Correspondence]
\label{Carter-Weyl-theorem}
Let $\sigma$ be an element of Weyl group of a simple complex finite-dimensional Lie algebra $\mathfrak g$ of rank $r$ such that $l(\sigma)=r$.  Let $\xi_1,\cdots,\xi_r$ be Carter roots for $\sigma$, with corresponding Carter matrix $K$.  Let $S_{\sigma}$ and $S_{K}$ denote the multi-sets of eigenvalues of $\sigma$ and $K$. 
\begin{enumerate}[{\normalfont (i)}]
\item
There are $\theta_1,\theta_2,\cdots, \theta_{r}$ in $[0,2\pi)$ such that 
\begin{align*}
S_{\sigma} &=\{e^{ i \theta_1},\cdots, e^{i \theta_{r}}\} \\[0.07in]
S_{K} &= \{ 2 - 2 \cos\left(\frac{\theta_{1}}{2} \right ) ,\cdots  ,  2 - 2 \cos\left(\frac{\theta_{r}}{2} \right )  \} 
\end{align*}
\item
There is an isomorphism between the left $(2 -2 \cos\left(\frac{\theta}{2} \right ))$-eigenspace of $K$ and the right  $e^{i\theta}$-eigenspace  of $\sigma$,  given by
$$(x_1,\cdots,x_r) \mapsto \sum_{j=1}^{r}  x_j e^{c_j i \theta/4 } \cdot \xi_j$$
\end{enumerate}
\end{thm}

For Coxeter elements,  the result goes back to the work of Coxeter \cite{COX} and Coleman \cite{COLE}.  We show how Carter's work \cite{CAR} can be used to generalize this to more general Weyl group elements.  Throughout, we closely follow the treatment by Fring-Liao-Olive in \cite{FLO} of the Coxeter case (see also the work of Brillon-Schechtman \cite{BS} for a similar discussion).  A fundamental tool employed in \cite{FLO} is that Dynkin diagrams can be bi-colored. This is clear, since these diagrams are trees.  Let us start by recalling in what way such a bi-colorability carries over to general Weyl group elements.

Given a product decomposition $\sigma=\sigma_{+}\sigma_{-}$ with $\sigma_{\pm}$ as in Equation (\ref{sigma-decomposition}), one associates a certain graph, called by Carter an admissible diagram.  The vertices correspond to the Carter roots and the $i$'th and $j$'th vertices are joined by $N_{i,j}$ lines where
$$N_{i,j}=K_{i,j}\cdot K_{j,i}$$
Hence, if the Carter roots are a set simple roots, one simply obtains the Dynkin diagram.  In very subtle work, Carter uses these diagrams in \cite{CAR} to classify conjugacy classes in Weyl groups. Note that a given conjugacy class might be described by different admissible diagrams.  

Contrary to the case of simple roots, in general there will be some indices $i\ne j$ with $K_{i,j}>0$.  It will be useful for our purposes to keep track of this:  Define a new graph exactly as the admissible diagrams,  except that if $K_{i,j}>0$ then the edges between the $i$'th and $j$'th vertex are dotted.  See the work of Stekolshchik \cite{STE} for a more in-depth study of such ``Carter diagrams''. 

\begin{example}
The conjugacy class $\textrm{E}_6(\textrm{a}_1)$ in the Weyl group of $\mathfrak e_6$ has admissible diagram
$$
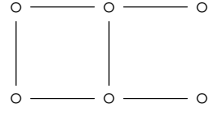
\xymatrix{
\circ \ar@{-}[r] \ar@{-}[d] & \circ  \ar@{-}[r] \ar@{-}[d]  & \circ  \\
\circ \ar@{-}[r] & \circ  \ar@{-}[r] & \circ
}$$
One of the underlying Carter diagrams is
$$
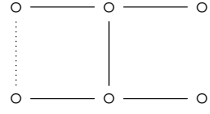
\xymatrix{
\circ \ar@{-}[r] \ar@{.}[d] & \circ  \ar@{-}[r] \ar@{-}[d]  & \circ  \\
\circ \ar@{-}[r] & \circ  \ar@{-}[r] & \circ
}$$
as can be seen by the following choices of roots:
$$\xymatrix{
\alpha_2\ar@{-}[r] \ar@{.}[d] &\alpha_3 + \alpha_4  \ar@{-}[r] \ar@{-}[d]  & \alpha_1  \\
\alpha_2+\alpha_4 \ar@{-}[r] & \alpha_5   \ar@{-}[r] & \alpha_6
}$$
where $\alpha_1,\cdots, \alpha_6$ are simple roots, indexed as in \cite{BOURBAKI}.  Different choices of roots can have different Carter diagrams with same admissible diagrams.
\end{example}
From the decomposition $\sigma= \sigma_{+}\sigma_{-}$ it follows that admissible diagrams (and Carter diagrams) can be bi-colored, simply by assigning one color to roots of the form $\xi_i^{+}$ and assigning the other color to roots of the form $\xi_{i}^{-}$.  So in the previous example one can color
$$
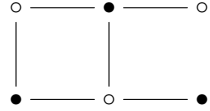
\xymatrix{
\circ \ar@{-}[r] \ar@{-}[d] & \bullet \ar@{-}[r] \ar@{-}[d]  & \circ  \\
\bullet \ar@{-}[r] & \circ  \ar@{-}[r] & \bullet
}$$   

This consequence of Carter's work allows to obtain the generalization of the 
Cartan-Coxeter correspondence described in Theorem \ref{Carter-Weyl-theorem}.  The following proof closely follows the seminal work by Fring-Liao-Olive \cite{FLO}, which in turn is based on Coxeter's crucial insights in \cite{COX}. 
We simply put here these arguments in their natural context of general bi-colorable root configurations.

\begin{proof}[Proof of Theorem \ref{Carter-Weyl-theorem}]

Let $\xi_{1},\cdots,\xi_{r}$ be the Carter roots, and group them as before into two groups  $\xi_{1}^{+},\cdots,\xi_{a}^{+}$ and $\xi_{1}^{-},\cdots,\xi_{b}^{-}$  of mutually orthogonal roots.
Let $r_{\xi_{i}}$ denote the reflection in the $i$'th Carter root, so
$$r_{\xi_{i}} \xi_j = \xi_j - K_{j,i} \xi_i$$
For an index $1\le i \le r$ we write $i \in +$ if $\xi_i$ is of the form $\xi_{j}^{+}$ for some $j$, and analogously we write $i \in -$ if $\xi_i$ is of the form $\xi_{j}^{-}$ for some $j$.

For $i \in \pm$,  clearly $\sigma_{\pm}\xi_{i}=-\xi_i$ and we now calculate $\sigma_{\mp} \xi_i$.  Recall from Equation (\ref{sigma-decomposition}) that the definition of $\sigma$ involves mutually orthogonal roots $\xi^{+}_{1},\cdots, \xi^{+}_{a}$ and mutually orthogonal roots $\xi^{-}_{1},\cdots, \xi^{-}_{b}$.  For notational convenience we set $t$ to be $b$ or $a$, depending on whether $i \in +$ or $i \in -$. Then 
\begin{eqnarray}
\sigma_{\mp} \xi_i = r_{\xi^{\mp}_{t}} \cdots  r_{\xi^{\mp}_{1}} \xi_{i}=r_{\xi^{\mp}_{t}} \cdots r_{\xi^{\mp}_{2}} (\xi_i-K_{i,\tilde 1}\xi_{1}^{\mp})
\end{eqnarray}
where $\tilde 1$ is the index such that $\xi_{1}^{\mp}=\xi_{\tilde 1}$. Iterating this process, one obtains for $i \in \pm$
\begin{eqnarray}
\label{opposite-equation}
\sigma_{\mp} \xi_i =\xi_{i} - \sum_{j \in \mp } K_{i,j}\xi_{j}
\end{eqnarray}
Consider the linear endomorphism of root space given by $\sigma_{+}+\sigma_{-}$.  As in \cite{BLM}, one obtains
\begin{align}
(\sigma_{+} + \sigma_{-})(\xi^{\pm}_i) &= - \sum_{j \in \mp} K_{i,j} \xi_j \\[0.07in]
\label{summation-equation}
& =\sum_{j=1}^{r} (2\delta_{i,j}-K_{i,j})\xi_{j}
\end{align}
Now let $x_1,\cdots, x_r$ denote the entries of a left eigenvector of the Carter matrix, so there is a scalar $c$ such that
\begin{eqnarray}
\label{eigenvector-equation}
\sum_{i=1}^{r} x_i K_{i,j}= c x_j
\end{eqnarray}
for all $j$.  If $j \in \pm$ then
\begin{eqnarray}
\label{partial-signed-equatio}
 \sum_{i \in \mp} x_i K_{i,j} =-2x_j + \sum_{i=1}^{r} x_i K_{i,j} = (c-2) x_j 
 \end{eqnarray}
Recall that $c_i = 1$ if $i\in +$ and $c_i=-1$ if $i\in -$. One obtains for $j \in \pm$
\begin{eqnarray}
\label{signed-eigenvector-equation}
\sum_{i=1}^{r} c_i \cdot  x_i  K_{i,j}= \pm 2x_j\mp \sum_{i \in \mp } x_i K_{i,j}=(\pm 2 \mp (c-2)) x_{j}=c_j \cdot (4-c)  \cdot x_j
\end{eqnarray}
Define
$$q_{\pm}=\sum_{i \in  \pm} x_i \xi_i$$
From Equation (\ref{opposite-equation}),  Equation (\ref{summation-equation}), and Equation (\ref{eigenvector-equation}) it follows that 
\begin{eqnarray*}
\sigma_{+}(q_{-})+\sigma_{-}(q_{+})&=& q_{-}+ \sum_{i \in -} x_i \sum_{j =1}^{r}(2 \delta_{i,j}-K_{i,j}) \xi_j +q_{+}+ \sum_{i \in +} x_i \sum_{j=1}^{r}(2 \delta_{i,j}-K_{i,j}) \xi_j \\
&=&q_{-}+q_{+}+\sum_{j=1}^{r} (2-c)  x_j \xi_j
\end{eqnarray*}
From Equation (\ref{opposite-equation}), Equation (\ref{summation-equation}), and Equation (\ref{signed-eigenvector-equation}) it follows that
\begin{eqnarray*}
\sigma_{+}(q_-)-\sigma_{-}(q_{+})&=&q_{-}+ \sum_{i \in -} x_i \sum_{j=1}^{r}(2 \delta_{i,j}-K_{i,j})\xi_j -q_{+}- \sum_{i \in +} x_i \sum_{j=1}^{r}(2 \delta_{i,j}-K_{i,j}) \xi_j \\
&=&q_{-}-q_{+} + \sum_{j=1}^{r} c_j \cdot (2-c) x_j \xi_j
\end{eqnarray*}
Therefore 
$$\sigma_{+}(q_{-})=q_{-}+ (2-c) \cdot q_{+}$$
and
$$\sigma_{-}(q_{+})=q_{+}+(2-c) \cdot q_{-}$$
Since clearly $\sigma_{\pm}(q_{\pm})=-q_{\pm}$, it follows that the action of $\sigma = \sigma_{+}  \sigma_{-}$ on the $q_{\pm}$-plane with respect to the basis $\{q_{+},q_{-}\}$ is described by the matrix
\begin{eqnarray}
\label{M-equation}
M=\begin{bmatrix}
(2-c)^2-1& c-2\\
2-c & -1
\end{bmatrix}
\end{eqnarray}
To proceed,  we relate the spectrum of the Carter matrix $K$ to the spectrum of the Weyl group element $\sigma$. The finite-order endomorphism $\sigma$ of root space is diagonalizable, and hence so is
\begin{align}
\label{square-equation}
(\sigma_{+} + \sigma_{-})^2=2+ \sigma_+ \sigma_{-} + \sigma_{-} \sigma_{+}=2+\sigma+\sigma^{-1} 
\end{align}
Let $\{\mu_1,\cdots,\mu_{r}\}$ denote the multi-set of eigenvalues of $\sigma$. Since the Killing form is symmetric and non-degenerate, the Carter matrix $K$ is symmetrizable and hence diagonalizable.  Let $\{\lambda_1,\cdots,\lambda_r\}$ denote its multi-set of eigenvalues.  We have shown earlier that for all $i$
$$(\sigma_{+} + \sigma_{-})(\xi_i)  =\sum_{j=1}^{r} (2\delta_{i,j}-K_{i,j})\xi_{j}$$
Hence,  $\sigma_{+} + \sigma_{-}$ is a diagonalizable endomorphism of root space.  Together with Equation (\ref{square-equation}) this yields two descriptions of the spectrum of $(\sigma_{+} + \sigma_{-})^2$,  and hence an equality of multi-sets
\begin{eqnarray}
\label{multiset-equation}
\{(2-\lambda_1)^2,\cdots,(2-\lambda_{r})^2 \} = \{2+\mu_{1}+\mu_{1}^{-1}, \cdots, 2+ \mu_{r} + \mu_{r}^{-1}\}
\end{eqnarray}
Choose an arbitrary $\mu_{j}$ and write this root of unity as $\mu_{j}= e^{ i \theta}$ with $\theta$ in $[0,2\pi)$.  Then
\begin{eqnarray}
\label{c-equation}
2+2\cos \theta = (2-c)^2
\end{eqnarray}
where $c$ is an eigenvalue of $K$.  The two roots of this quadratic equation for 
$c$ are $2\pm 2\cos (\theta/2)$ and we now know that at least one of them is an 
eigenvalue.  By Equation (\ref{signed-eigenvector-equation}), if $c$ is an 
eigenvalue of $K$, so is $4-c$. It follows that in fact $2\pm 2\cos (\theta/2)$ are both 
eigenvalues of $K$. If $\mu_{j} \ne \pm 1$, then $\overline{\mu}_{j}=\mu_{j}^{-1}
=e^{i (2\pi -\theta)}$ is a second eigenvalue and
$$\{2+2\cos (\theta/2),2- 2\cos (\theta /2)\} =\{2-2\cos (\theta/2),2- 2\cos ((2\pi - \theta )/2)\}$$
is as predicted by the theorem. If $\mu_j=e^{i\pi}=-1$, then Equation (\ref{c-equation}) implies that $2=2-2\cos(\pi/2)$ is an eigenvalue of $K$.  Finally, $\mu_j=1$ does not happen: This follows from our assumption $l(\sigma)=r$, by comparing with the eigenvalues listed by Carter in \cite{CAR} (Table 3).  This proves the eigenvalue part of the Carter-Weyl correspondence.  

Consider now to the comparison of the eigenvectors. Let $e^{i\theta}$ be an eigenvalue of $\sigma$ and choose $c = 2 - 2 \cos(\theta/2)$ for the eigenvalue of the Carter matrix. With this choice of $c$, the eigenvalues of the matrix $M$ in Equation (\ref{M-equation}) are $e^{\pm i\theta}$.  Calculating the right eigenvector of $M$ corresponding to $e^{i \theta}$ yields the right eigenvector of $\sigma$ given by
\begin{eqnarray}
\label{Lambda-equation}
\Lambda=e^{i \frac{\theta}{4}} q_{+} + e^{-i\frac{\theta}{4}} q_{-}= \sum_{j=1}^{r} x_j e^{c_j i \frac{\theta}{4}} \cdot \xi_j
\end{eqnarray}

As shown before,  the Carter roots $\xi_j$ are linearly independent.  It follows that for each fixed eigenvalue of $K$,  linearly independent eigenvectors of the Carter matrix yield via Equation (\ref{Lambda-equation}) linearly independent eigenvectors for $\sigma$: The multiplying factor $e^{c_j i \frac{\theta}{4}}$ is constant for each $j$.  Hence,  one obtains the claimed isomorphism of eigenspaces in the Carter-Weyl correspondence.
\end{proof}
A crucial consequence of the theorem is the determination of the inner products of $\Lambda$ with all the Carter roots. 
\begin{cor}
\label{main-corollary}
Let $\Lambda$ in $\mathfrak h$ be the eigenvector with eigenvalue $e^{i\theta}$ corresponding to $(x_1,\cdots,x_{r})$ under the Carter-Weyl correspondence.  Then, after scaling $\Lambda$, one obtains for all $j$
$$\Lambda \cdot \xi_j = c_j   e^{-c_{j}i \frac{\theta}{4}} \cdot \xi_{j}^{2} x_j$$
In particular,  the values $|\Lambda \cdot \xi_j|$ yield the absolute values of a  right eigenvector of the Carter matrix. 
\end{cor}
\begin{proof}
Let $t=e^{i \frac{\theta}{4}}$. Using Equation (\ref{partial-signed-equatio}) and Equation (\ref{Lambda-equation}),  one obtains for $j \in \pm$ 
\begin{align*}
\Lambda \cdot \xi_j &=  \xi_{j}^{2}  x_j t^{\pm 1}+ \frac{\xi_{j}^2 t^{ \mp 1}}{2}  \cdot \sum_{s \in \mp } x_s \cdot K_{s,j}  \\[0.07in]
&= \left( t^{\pm 1}-(t^{2}+t^{-2}) t^{\mp 1}/2\right ) \cdot \xi_{j}^{2} x_{j}\\[0.07in]
&=  i\sin \left (\frac{ \theta}{2} \right ) c_j e^{-c_{j} i \frac{\theta}{4}} \cdot  \xi_{j}^{2} x_j 
\end{align*}
This implies the first part of the corollary.  After scaling $\Lambda$ one obtains 
$|\Lambda \cdot \xi_j| = \left | \xi_{j}^2 x_j \right |$.  Hence,  whereas the $x_j$'s are the entries of a left eigenvector of the Carter matrix, the absolute values $|\Lambda \cdot \xi_j|$ are the absolute values of the entries of a right eigenvector. 
\end{proof}

In principle, the Carter-Weyl correspondence allows to answer Question \ref{main-question} and in particular the calculation of the Toda-Weyl mass spectrum: Corollary \ref{main-corollary} describes how $\Lambda$ pairs with a collection of Carter roots. Since Carter roots give a basis of root space, this determines how $\Lambda$ pairs with any root, in particular with the orbit representatives $\gamma_i$.  The corresponding formulas will be particularly nice if there is a close relation between the Carter roots and the orbit representatives.  

If $\sigma$ is a Coxeter element,  Kostant shows the following in \cite{KOS}: Let $\alpha_1,\cdots,\alpha_{r}$ be a collection of simple roots and as before let $c_i=\pm 1$ according to a choice of bi-coloration of the simple roots. Consider the Coxeter element 
$$\sigma=\prod_{c_{i}=1} r_{\alpha_{i}} \prod_{c_{i}=-1} r_{\alpha_{i}}$$
Kostant showed that the number of positive roots that get mapped via $\sigma$ to a negative root is the rank $r$ of $\mathfrak g$.  Further,  each orbit of $\langle \sigma \rangle$ on the set of roots contains exactly one of these ``sign-changer'' roots.  Using the fact that off-diagonal entries of the Cartan matrix are non-positive one can show, see \cite{BS} (Lemma 3.9), that $\pm \alpha_i$ is either a sign-changer or the image under $\sigma$ of a sign-changer. Further,  if the sign is chosen according to the bi-coloration of the Dynkin diagram, then this dichotomy corresponds exactly to the color of each root.  It follows that all $c_i \alpha_i$ lie in distinct orbits and in fact this gives all the orbits, since the number of roots is $hr$ where $h$ is the Coxeter number.  Since for any root $\beta$ one has $r_{-\beta}=r_{\beta}$,  it follows that there is a choice of Carter roots that are in fact a collection of orbit representatives. 

Some care is required to obtain generalizations of this relation for other Weyl group elements. To illustrate the point we discuss the conjugacy class $\textrm{E}_{6}(a_1)$. 

Let $\alpha_1,\cdots,\alpha_6$ be a choice of simple roots, indexed as in \cite{BOURBAKI}.  One possible representative of the $\textrm{E}_{6}(a_1)$ Carter diagram is
$$\xymatrix{
\alpha_2\ar@{.}[r] \ar@{-}[d] & \alpha_2+ \alpha_3 + \alpha_4  \ar@{-}[r] \ar@{-}[d]  & \alpha_1  \\
\alpha_4 \ar@{-}[r] & \alpha_5   \ar@{-}[r] & \alpha_6
}$$
The corresponding Weyl group element is
$$\sigma=r_{\alpha_{1}}r_{\alpha_{2}} r_{\alpha_{5}} r_{\alpha_4} r_{\alpha_{6}} r_{\alpha_{2}+\alpha_{3}+\alpha_{4}}$$
It turns out that the signed Carter roots (where the sign is chosen according to a bi-coloration of the Carter roots) do not lie in distinct orbits: Otherwise, $\alpha_2$ and $-\alpha_4$ should lie in distinct orbits, however,  the orbit of $\alpha_2$ is given by

\begin{align*}
 \{ &\alpha_{2}, -\alpha_{1}-\alpha_{3},-\alpha_{2}-\alpha_{3}-\alpha_{4}-\alpha_{5}, -\alpha_{2}-\alpha_{4}-\alpha_{5}-\alpha_{6}, -\alpha_{4}, \alpha_{2}+\alpha_{4}+\alpha_{5},\\[0.08in]
&\alpha_{4}+\alpha_{5}+\alpha_{6}, \alpha_{1}+\alpha_{3}+\alpha_{4}, \alpha_{3}\}
\end{align*}

Nonetheless, there is a different representative of the conjugacy class $\textrm{E}_{6}(a_1)$ that has a better relation between Carter roots and orbit structure.  Consider a second representative of the Carter diagram,  given by
$$\xymatrix{
\alpha_2\ar@{-}[r] \ar@{.}[d] &\alpha_3 + \alpha_4  \ar@{-}[r] \ar@{-}[d]  & \alpha_1  \\
\alpha_2+\alpha_4 \ar@{-}[r] & \alpha_5   \ar@{-}[r] & \alpha_6
}$$
and define 
\begin{align}
\label{another-equation}
\sigma=r_{\alpha_1}r_{\alpha_{2}} r_{\alpha_{5}} r_{\alpha_{6}} r_{\alpha_{2}+\alpha_{4}} r_{\alpha_{3}+\alpha_{4}}
\end{align}
A direct calculation yields the following eight orbits $\mathcal O_{1},\cdots, \mathcal O_{8}$ for the action of the cyclic group $\langle \sigma \rangle$ on the set of roots: 

\begin{align*}
\mathcal O_1 &= &\{ & \alpha_1, \alpha_2+\cdots +\alpha_5, \alpha_1+\alpha_3+\cdots+\alpha_6, \alpha_2+\alpha_3+2\alpha_4+\alpha_5, \alpha_6,-\alpha_5-\alpha_6,\\[0.08in]
&  & &-\alpha_1-\alpha_2-\alpha_3-2\alpha_4-\alpha_5,-\alpha_2-\cdots - \alpha_6,-\alpha_1-\alpha_3-\alpha_4\}\\[0.08in]
\mathcal O_2&= &\{  &\alpha_2,  \alpha_1+\alpha_3,\alpha_3+\alpha_4+\alpha_5,\alpha_4+\alpha_5+\alpha_6,\alpha_2+\alpha_4,
-\alpha_4-\alpha_5,\\[0.08in]
& & &-\alpha_2-\alpha_4-\alpha_5-\alpha_6,-\alpha_1-\cdots-\alpha_4,-\alpha_3\} \\[0.08in]
\mathcal O_3&=&\{ &\alpha_4,-\alpha_1-\alpha_3-\alpha_4-\alpha_5,-\alpha_2-\alpha_3-2\alpha_4-2\alpha_5-\alpha_6,\\[0.08in]
& & &-\alpha_1-\alpha_2-\alpha_3-2\alpha_4-\alpha_5-\alpha_6,
-\alpha_2-\alpha_3-\alpha_4,\alpha_2+\alpha_4+\alpha_5,\\[0.08in]
& & &\alpha_1+\cdots+\alpha_6,\alpha_1+\alpha_2+2\alpha_3+2\alpha_4+\alpha_5,\alpha_3+\alpha_4+\alpha_5+\alpha_6\} \\[0.08in]
\mathcal O_4 &=& \{ &\alpha_5,\alpha_1+\alpha_2+\alpha_3+2\alpha_4+2\alpha_5+\alpha_6,\alpha_1+2\alpha_2+2\alpha_3+3\alpha_4+2\alpha_5+\alpha_6,\\[0.08in]
& & & \alpha_1+\alpha_2+2\alpha_3+2\alpha_4+\alpha_5+\alpha_6,\alpha_3+\alpha_4,-\alpha_1-\cdots-\alpha_5,\\[0.08in]
& & &  -\alpha_1-\alpha_2-2\alpha_3-2\alpha_4-2\alpha_5-\alpha_6,
 -\alpha_1-\alpha_2-2\alpha_3-3\alpha_4-2\alpha_5-\alpha_6,\\[0.08in]
 & & &-\alpha_2-\alpha_3-2\alpha_4-\alpha_5-\alpha_6\}
 \end{align*}
 and $\mathcal O_{5} = -\mathcal O_{1}$, $\mathcal O_{6} = -\mathcal O_{2}$, $
\mathcal O_7 = - \mathcal O_3$, $\mathcal O_8 = - \mathcal O_4$.  Since $\sigma$ from Equation (\ref{another-equation}) can be written as
$$\sigma=r_{\alpha_1}r_{\alpha_{2}} r_{\alpha_{5}} r_{-\alpha_{6}} r_{-(\alpha_{2}+\alpha_{4})} r_{-(\alpha_{3}+\alpha_{4})}$$
it follows that there are Carter roots for $\sigma$ that lie in distinct orbits: $\xi_1= \alpha_1  , \;\; \xi_2 = \alpha_2  ,\;\; \xi_3= \alpha_5  , \; \; \xi_4= - \alpha_6 , \;\; \xi_5 = - (\alpha_2 + \alpha_4) , \;\; \xi_6 = -(\alpha_3 + \alpha_4)$. 
Having illustrated some of the subtle relations between Carter roots and orbit representatives, we focus for the remainder of this work on an especially interesting family of conjugacy classes:

The Weyl group conjugacy class of Coxeter elements is an example of a regular primitive conjugacy class.  Such conjugacy classes are known to be the building blocks of all regular conjugacy classes.  Apart from the Coxeter class, there are in fact not many others for complex simple Lie algebras: Apart from a few cases for the exceptional algebras, there is exactly one infinite family of regular primitive non-Coxeter conjugacy classes.  In the notation by Carter \cite{CAR},  it is denoted by $\textrm{D}_{2n}(a_{n-1})$.  

\section{Toda-Weyl mass spectrum for $\textrm{D}_{2n}(a_{n-1})$}
\label{example-section}
In this section we prove Theorem \ref{type-D-theorem}.  A representative for the conjugacy class $\textrm{D}_{2n}(a_{n-1})$ is given by Carter in \cite{CAR} (Proposition 25). We supplement this with the description of the orbit structure. Throughout, our indexing convention for simple roots is as in \cite{BOURBAKI}.
\begin{lem}
\label{orbit-lemma}
Consider the element $\sigma=\sigma_{+} \sigma_{-}$ in the Weyl group of type $\textrm{\emph{D}}_{2n}$ where
$$\sigma_{+}= r_{\alpha_{1}}r_{\alpha_{3}} \cdots r_{\alpha_{2n-1}} r_{\beta }\;\;,\;\; \sigma_{-}=r_{\alpha_{2}} r_{\alpha_{4}} \cdots r_{\alpha_{2n-2}}$$
for $\beta=\alpha_{n}+2\sum_{i=n+1}^{2n-2} \alpha_i + \alpha_{2n-1}+\alpha_{2n}$.  Then $\sigma$ is a representative of the conjugacy class $\textrm{\emph{D}}_{2n}(a_{n-1})$, the number of orbits of $\langle \sigma \rangle$ acting on the set of roots is $4n-2$, and a list of orbit representatives is
\begin{eqnarray*}
&& \alpha_1 \\[0.07in]
&& \vdots \\[0.07in]
&&\alpha_{2n-1} \\[0.07in]
&& \beta \\[0.07in]
&& \eta_i := \sum_{t=n-i}^{n} \alpha_{t} \;\;\; \textrm{ for } 1\le i \le n-1 \\[0.07in]
&& \nu_i:= \sum_{t=n-i}^{n} \alpha_{t} + \alpha_{n+1} \;\;\; \textrm{ for } i=0 \textrm{ as well as } 2\le i \le n-1
\end{eqnarray*}
\end{lem}
\begin{proof}
It will be convenient to express the Weyl group elements in terms of permutations.  From this perspective the roots are $\pm e_{i} \pm e_{j}$ and the Weyl group elements acts via signs and permutations.  Note that $\beta = e_{n}+e_{n+1}$ and the corresponding reflection is given by
$$r_{\beta}=(n \; , \;  -(n+1))$$
For a cycle $(x_1, \cdots, x_a, -x_1,\cdots,-x_a)$ we write for simplicity $(x_1, \cdots, x_a)^{-}$.  If $n$ is odd, define
\begin{align*}
\sigma_{1} &=(1\; ,\; 2 \; , \;4 \; ,  \; \cdots \; , \;  (n-1) \; ,\;  -n \; , \; -(n-2) \; , \; \cdots \; , \; -3 )^{-} \\[0.07in]
\sigma_{2} &= (n+1 \;, \; n+3 \; , \; \cdots \; , \; 2n \;,  \; 2n-1 \;, \; 2n-3 \;, \; \cdots \;,  n+2)^{-}
\end{align*}
If $n$ is even, define
\begin{align*}
\sigma_{1} &= (1\; ,\; 2 \; , \;4 \; ,  \; \cdots \; , \;  n \; ,\;  - (n-1) \; , \; -(n-3) \; , \; \cdots \; , \; -3 )^{-} \\[0.07in]
\sigma_{2} &=(n+1 ,  -(n+2)  ,  -(n+4) , \cdots  ,  -2n ,   -(2n-1) ,  -(2n-3) ,  \cdots ,   -(n+3))^{-}
\end{align*}
An explicit calculation yields the disjoint cycle decomposition $$\sigma= \sigma_1 \sigma_2$$ 
In particular,  in the notation of \cite{CAR} (Section 7), the signed cycle-type of $\sigma$ is $[\overline{n} \; \overline{n}]$. By \cite{CAR} (Proposition 25),  this corresponds to the conjugacy class $\textrm{D}_{2n}(a_{n-1})$, as desired.

We now compute orbit representatives.  As discussed by Reeder \cite{REE},  for regular conjugacy classes it follows from work of Springer \cite{SPR} (Proposition 4.1) that every orbit of $\sigma$ has exactly $\textrm{ord}(\sigma)$ elements.  The  total number of roots is known to be $h \cdot \textrm{rank }\mathfrak g$ where $h=4n-2$ is the Coxeter number. It follows that the number $s$ of orbits of $\sigma$ is given by 
$$\frac{h \cdot r}{\textrm{ord}(\sigma)}=\frac{(4n-2)\cdot 2n}{2n}=4n-2$$
Hence it suffices to show that the Carter roots and $\eta_i$'s and $\nu_i$'s all lie in distinct orbits.  We write $\gamma \sim \delta$ if two roots $\gamma$ and $\delta$ are in the same orbit.  Note that $\eta_i=e_{n-i}-e_{n+1}$ and $\nu_i=e_{n-i}-e_{n+2}$.

From the explicit form of $\sigma_1$ and $\sigma_2$ it follows that
\begin{align}
\label{sigma-equation}
| \sigma^{k}(j)| \le n
\end{align}
for all $k$ and all $1\le j <n$, and
\begin{align}
\label{sigma-equation2}
| \sigma^{k}(j)| > n
\end{align}
for all $k$ and all $n<j\le 2n$.  Suppose now $\sigma^{k} \eta_i =\eta_j$.  From Equation (\ref{sigma-equation}) it follows that $\sigma_{2}^{k}(n+1)=n+1$. Hence $k$ is $0$ modulo $2n$ and $i=j$.  In the same manner, replacing $n+1$ by $n+2$,  one obtains $\nu_{i} \not \sim \nu_{j}$ if $i\ne j$.   Suppose now for contradiction that $\eta_{i} \sim \nu_{j}$.  Looking at the action on $n+1$ this implies $\sigma^{n-1}\eta_{i} = \nu_{j}$ if $n$ is odd and $\sigma^{n+1}\eta_{i} = \nu_{j}$ if $n$ is even.  Therefore $\sigma^{n\pm1}(n-i)$ would need to be of the form $n-j$ with either $j=0$ or $2\le j \le n-1$. However, the only $1\le a \le n-1$ such that $\sigma^{n\pm 1}(a)$ is positive, is $a=n-1$ with $n$ even, and $\sigma^{n-1}(n-1)=n$.  It follows that $\eta_{i} \not \sim \nu_{j}$ for all $i,j$.  

By Equation (\ref{sigma-equation}) and Equation (\ref{sigma-equation2}),  if $
\alpha_i=e_{i}-e_{i+1}$ is in the orbit of any $\eta_j$ or $\nu_{j}$ one needs $i=n
$.  Looking at the action on $n+1$,  clearly $\alpha_n \not \sim \eta_i$ for every $i
$.  If $\alpha_n \sim \nu_j$, then looking at the action on $n+1$, one sees that $
\sigma^{n-1}\alpha_n= \nu_j$ when $n$ is odd, and $\sigma^{n+1} \alpha_n=
\nu_j$ when $n$ is even. However, if $n$ is odd then $\sigma^{n-1}(n)=-(n-1)$ 
and if $n$ is even $\sigma^{n+1}(n)=n-1$.  It follows that indeed $\alpha_n$ is 
not in the orbit of any $\eta_j$ or $\nu_{j}$.

Consider now $\beta=e_{n}+e_{n+1}$. As for $\alpha_n$, one obtains $\beta \not \sim \eta_i$. If $\beta \sim \nu_j$, then looking at the action on $n+1$, one sees that that $\sigma^{-1}\beta= \nu_j$ when $n$ is odd, and $\sigma \beta=
\nu_j$ when $n$ is even.  If $n$ is odd, then $\sigma^{-1} \beta = -e_{n-1}-e_{n+2} \ne \nu_j$. If $n$ is even, then also $\sigma \beta =-e_{n-1}-e_{n+2} \ne \nu_j$.  In conclusion, $\beta$ is not in the same orbit as any of the $\eta_{i}$'s or $\nu_{j}$'s.  Furthermore,  by Equation (\ref{sigma-equation}) and Equation (\ref{sigma-equation2}),  $\beta \not \sim \alpha_i$ for $i\ne n$. Looking at the action on $n$, one also sees that there is no $k$ with $\sigma^{k} \beta= \alpha_n$.

To complete the proof of the Lemma, it remains to show that all the $\alpha_{i}$'s with $1\le i \le 2n-1$ lie in distinct orbits.  If $i <n$ and $j>n$ it follows from Equation (\ref{sigma-equation}) that $\alpha_{i} \not \sim \alpha_{j}$.  For $1\le a,b \le n$ there is a unique $k$ modulo $2n$ such that $\sigma^{k}(a)=b$ and we define $d(a,b):=k$.  As $i$ varies from $1$ to $n-1$ the values of $d(i, -(i+1))$ are exactly the $n-1$ values  $n+1 \; , \; n-2 \; , \; n+3 \; , \; n-4 \; , \; \cdots \; ,\; 3 \; , \; 2n-2 \; , 1$ if $n$ is odd,  and $n+1 \; , \; n-2 \; , \; n+3 ,\;   \; n-4 \; , \; \cdots \; , \; 2 \; , \; 2n-1$ if $n$ is even. In particular, if $i\ne j$ 
$$d(i,-(i+1)) \ne  d(j,-(j+1))$$ 
Since $d(\sigma^{t}a,\sigma^{t}b)=d(a,b)$ for all $1\le a,b \le n$ and all $t$, it follows that $\alpha_{i} \not \sim \alpha_{j}$ for $1\le i,j \le n-1$.  Suppose now $n+1 \le i,j < 2n$.  Then the values of $d(i,-(i+1))$ are exactly the $n-1$ values $2n-1 \; , \; 2 \; , 2n-3 \; , 4 \; , \; \cdots \; , n+2 \; , \; n-1 $ if $n$ is odd, and $1 \; , \; 2n-2 \; , \; 3 \; , \; 2n-4 \; , \; \cdots \; , \; n+2 \; ,\; n-1$ if $n$ is even.  As before,  one obtains $\alpha_{i} \not \sim \alpha_{j}$. Finally consider $\alpha_{n}=e_{n}-e_{n+1}$.  By Equation (\ref{sigma-equation}) and Equation (\ref{sigma-equation2}) it follows that $\alpha_{n} \not \sim \alpha_i$ for $i \ne n$.
\end{proof}

Having established the orbit structure,  we now answer Question \ref{main-question} for the conjugacy class $\textrm{D}_{2n}(a_{n-1})$.

\begin{proof}[Proof of Theorem \ref{type-D-theorem}]
Suppose the theorem is known for $\sigma$,  and $\sigma'=\mu\sigma \mu^{-1}$ is another element in the conjugacy class.  If $\gamma$ is an orbit representative for $\sigma$, then $\mu \gamma$ is an orbit representative for $\sigma'$, if $\Lambda$ is an eigenvector for $\sigma$ then $\mu \Lambda$ is a corresponding eigenvector for $\sigma'$. Since the Killing form is invariant under the Weyl group action it follows that $\mu \Lambda \cdot \mu \gamma=\Lambda \cdot \gamma$.  Hence, it suffices to prove the theorem for one element in the conjugacy class $\textrm{D}_{2n}(a_{n-1})$.  We therefore assume that $\sigma=\sigma_{+} \sigma_{-}$, where as in Lemma \ref{orbit-lemma}
\begin{align*}
\sigma_{+} &= r_{\alpha_{1}}r_{\alpha_{3}} \cdots r_{\alpha_{2n-1}} r_{\beta}\\[0.07in]
\sigma_{-}&=r_{\alpha_{2}} r_{\alpha_{4}} \cdots r_{\alpha_{2n-2}}
\end{align*}
for $\beta=\alpha_{n}+2\sum_{i=n+1}^{2n-2} \alpha_i + \alpha_{2n-1}+\alpha_{2n}$. The corresponding Carter diagram is given by
$$ 
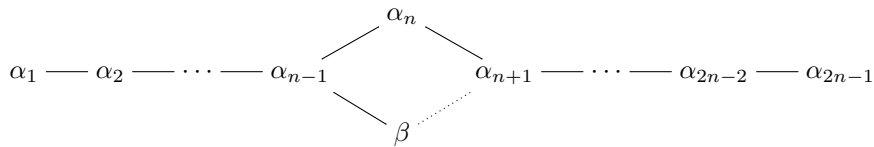
\xymatrix @R=0.8pc @C=1.3pc {
&&& & \alpha_{n} \ar@{-}[dl]  \ar@{-}[dr]   & & &&& \\
\alpha_1 \ar@{-}[r] &\alpha_2 \ar@{-}[r] &  \cdots \ar@{-}[r]& \alpha_{n-1} &  & \alpha_{n+1}  \ar@{-}[r] & \cdots \ar@{-}[r] & \alpha_{2n-2} \ar@{-}[r] & \alpha_{2n-1} &  \\
&&&&\beta \ar@{-}[ul] \ar@{.}[ur] & &&&&
}$$

We order the Carter roots $\xi_1,\cdots, \xi_{2n}$ as $\alpha_1, \cdots, \alpha_n, \beta, \alpha_{n+1}, \cdots, \alpha_{2n-1}$.  Recall also that the roots involved in defining $\sigma_{+}$ have $c_i=1$ and the roots involved in $\sigma_{-}$ have $c_i=-1$.  One has
$$\eta_{i} = \sum_{t=n-i}^{n} \xi_{t}$$ 
as well as
$$\nu_{0}=\xi_{n}+\xi_{n+2} \;\; , \;\; \nu_{i}=\xi_{n+2}+ \sum_{t=0}^{i} \xi_{n-t} \;\;\; (\textrm{for } 2\le i \le n-1)$$

For $\zeta=e^{2\pi i /2n}$ we now describe the (two-dimensional) $\zeta$-eigenspace of $\sigma$ using the Carter-Weyl correspondence.  All entries of the Carter matrix $K$ can be read off the Carter diagram in complete analogy with the relation between Cartan matrices and Dynkin diagrams.  For example for $n=4$: 

$$K=\begin{bmatrix}
2&-1&0&0&0&0&0&0\\
-1&2&-1&0&0&0&0&0\\
0&-1&2&-1&-1&0&0&0\\
0&0&-1&2&0&-1&0&0\\
0&0&-1&0&2&1&0&0\\
0&0&0&-1&1&2&-1&0\\
0&0&0&0&0&-1&2&-1\\
0&0&0&0&0&0&-1&2\\
\end{bmatrix}$$

Let $t=e^{i \pi/4n}$.  Since the Carter matrix $K$ is symmetric, the task is to describe the right eigenvectors of $K$ with eigenvalue $2-2\cos(\frac{\pi}{2n
})=2-(t^2+t^{-2})$.  By \cite{BOU} (Table 1),  together with the Carter-Weyl correspondence,  we know that this eigenspace is two-dimensional.  Let $(x_1,\cdots,x_{2n})^{\textrm{T}}$ be an eigenvector.  Assume for the moment that $x_1\ne 0$, so after scaling we can set $x_1=\sin \left (\frac{\pi}{2n} \right )$.  The first row of the eigenvector equation yields
$$x_2=(t^2+t^{-2})x_1=\sin\left( \frac{2\pi}{2n}\right)$$
In this manner (exactly as for Lie algebras of type A) one obtains for all $1\le j \le n-1$
$$x_j = \sin \left (\frac{j \cdot \pi}{2n} \right )$$
From the way $\alpha_{n-1}$ pairs with the other Carter roots,  the eigenvector equation yields
$$x_n+x_{n+1}=1$$
Analogously, from $\alpha_n$, $\beta$, $\alpha_{n+1}$ respectively,  one obtains

\begin{align*}
x_{n+2}&=-x_{n-1}+(t^2+t^{-2})x_n\\[0.08in]
x_{n+2}&=x_{n-1}-(t^2+t^{-2})x_{n+1}\\[0.08in]
x_{n+3}&=-x_{n}+x_{n+1}+(t^2+t^{-2})x_{n+2}
\end{align*}
Let us make the Ansatz $x_{2n}=\sin\left(\frac{\pi}{2n} \right)$. By an analogous argument as before,  one obtains for all $n+2 \le i \le 2n$
$$x_i  = \sin\left(\frac{2n-i+1}{2n} \cdot \pi \right)$$
Then $x_n=1$ and $x_{n+1}=0$ satisfy all constraints and this yields the eigenvector
\begin{align*}
\textbf{v}_{1}&= (\sin \left (\frac{\pi}{2n} \right ), \cdots,\sin \left (\frac{(n-1) \pi}{2n} \right ),1,0,\sin \left (\frac{(n-1) \pi}{2n} 
\right ),\cdots ,\sin \left (\frac{\pi}{2n} \right ))^{\textrm{T}}
\end{align*}
Now make the Ansatz $x_{2n}=0$.  The eigenvector equation yields for all $n+2 \le i \le 2n$
$$x_i=0$$
In this case $x_n=x_{n+1}=1/2$ satisfy all constraints and this yields the eigenvector
\begin{align*}
\textbf{v}_{2} &= (\sin \left (\frac{\pi}{2n} \right ), \cdots,\sin \left (\frac{(n-1) \pi}{2n} \right ),\frac{1}{2},\frac{1}{2},0,\cdots ,  
0 )^{\textrm{T}}
\end{align*}
Since we know the eigenspace is $2$-dimensional, $\textbf{v}_1$ and $\textbf{v}_{2}$ is a basis.

Let $\Lambda$ correspond to $\textbf{v}_{1}$ under the Carter-Weyl correspondence. We now calculate the corresponding Weyl-Orbit values.  Since by Lemma \ref{orbit-lemma} all Carter roots $\xi_i$ lie in distinct orbits, we obtain from Corollary \ref{main-corollary} the $2n$ Weyl-Orbit values corresponding to the absolute values of the entries of $\textbf{v}_{1}$.

We claim that for $1\le k \le n-1$
\begin{align}
\label{pairing-equation}
\Lambda \cdot \eta_k =\frac{1}{2i}(t^{(-1)^{n-k+1}(2(n-k)-1)}-t^{(-1)^{n}(2n+1)})
\end{align}
For $k=1$ it follows from Corollary \ref{main-corollary} that
\begin{align*}
\Lambda \cdot \eta_1 &=\Lambda_{+} \cdot (\xi_{n-1}+\xi_{n}) \\[0.1in]
&=\frac{1}{2i} ((-1)^{n}t^{(-1)^{n-1}}(t^{2(n-1)}-t^{-2(n-1)})+(-1)^{n+1}t^{(-1)^{n}}(t^{2n}-t^{-2n})) \\[0.1in]
&=(-1)^{n} \frac{1}{2i}(-1)^{n}(t^{(-1)^{n}(2n-3)}-t^{(-1)^{n}2n+1})
\end{align*}
as predicted.  Suppose now $\Lambda_{+} \cdot \eta_{k}$ is as in Equation (\ref{pairing-equation}).  Then 
\begin{align*}
\Lambda \cdot \eta_{k+1} &=\Lambda_{+} \cdot (\xi_{n}+ \cdots + \xi_{n-k-1})\\[0.1in]
&=\Lambda_{+}\cdot \eta_k + \Lambda_{+} \cdot \xi_{n-k-1}\\[0.1in]
&=\frac{1}{2i}(t^{(-1)^{n-k+1}(2(n-k)-1)}-t^{(-1)^{n}(2n+1)}\\[0.1in]
&\;\;\; \;\;+(-1)^{n-k}t^{(-1)^{n-k+1}}(t^{2(n-k-1)}-t^{-2(n-k-1)}))\\[0.1in]
&=\frac{1}{2i}(t^{(-1)^{n-k}(2n-2k-3)} -t^{(-1)^{n}(2n+1)})
\end{align*}
Hence,  Equation (\ref{pairing-equation}) holds for all $1\le k \le n-1$.  Note that 
$$\left |\frac{t^{a}-t^{b}}{2i} \right |= \left |t^{-(a+b)/2}  \cdot \frac{t^{a}-t^{b}}{2i} \right |= \left |\sin\left ( \frac{a-b}{8n} \cdot \pi \right ) \right |$$
Let 
$$c_{n,k}= \begin{cases}
0 \;\;\; \textrm{ if $k\equiv n \mod 2$ }\\[0.1in]
(-1)^{n+1} \;\;\; \textrm{if $n \not \equiv k \mod 2$}
\end{cases}$$ 
It follows from Equation (\ref{pairing-equation}) that
\begin{align*}
\left | \Lambda \cdot \eta_{n-k} \right | = \left| \sin \left( \frac{(-1)^{n+1}n+(-1)^{k+1}k +c_{n,k}}{4n} \cdot \pi \right) \right |
\end{align*}
If $n$ is even, these absolute values as $k$ varies from $1$ to $n-1$ are
$$\sin\left( \frac{n}{4n} \cdot \pi \right ), \sin\left( \frac{n+2}{4n} \cdot \pi \right ),\sin\left( \frac{n-2}{4n} \cdot \pi \right ),\sin\left( \frac{n+4}{4n} \cdot \pi \right ),\sin\left( \frac{n-4}{4n} \cdot \pi \right)$$
$$\cdots, \sin \left(\frac{2n-2}{4n} \cdot \pi \right), \sin \left(\frac{2}{4n} \cdot \pi \right)$$
If $n$ is odd, one obtains the absolute values
$$\sin\left( \frac{n+1}{4n} \cdot \pi \right ), \sin\left( \frac{n-1}{4n} \cdot \pi \right ),\sin\left( \frac{n+3}{4n} \cdot \pi \right ),\sin\left( \frac{n-3}{4n} \cdot \pi \right ),  \cdots$$
$$\cdots, \sin \left(\frac{2n-2}{4n} \cdot \pi \right), \sin \left(\frac{2}{4n} \cdot \pi \right)$$
Hence, in both cases
$$\{|\Lambda \cdot \eta_{k} | \; \big |  \; 1\le k \le n-1 \} = \{\sin \left(\frac{\pi}{2n} \right ), \cdots, \sin \left(\frac{(n-1)\pi}{2n} \right )\}$$
We now calculate $|\Lambda \cdot \nu_{i}|$. For $2\le i \le n-1$ 
\begin{align*}
\Lambda \cdot \nu_{i} &=\Lambda \cdot \eta_{i} + \Lambda \cdot \alpha_{n+1} \\[0.07in]
& = \Lambda\cdot \eta_{i} +(-1)^{n} t^{(-1)^{n+1}} \cdot \frac{t^{2(n-1)}-t^{-2(n-1)}}{2i} 
\end{align*}
Therefore
$$\Lambda \cdot \nu_{n-k}= \frac{t^{(-1)^{k+1}(2k-1)}-t^{(-1)^{n}(2n+1)}+(-1)^{n} t^{(-1)^{n+1}}(t^{2(n-1)}-t^{-2(n-1)}) }{2i}$$
Since $t^{4n}=-1$, $\left | \Lambda \cdot \nu_{n-k} \right|$ is given by
\begin{align*}
 & \left  |t^{(-1)^{n+1}2n}    \frac{t^{(-1)^{k+1}(2k-1)}-t^{(-1)^{n}(2n+1)}+(-1)^{n} t^{(-1)^{n+1}}(t^{2(n-1)}-t^{-2(n-1)})} {2i} \right | \\[0.17in]
= &\left| \frac{t^{(-1)^{k+1}(2k-1)+(-1)^{n+1}2n}-t^{(-1)^{n}} +(-1)^{n}(t^2+t^{-2})(-1)^{n}t^{(-1)^{n+1}}}{2i} \right |   \\[0.17in]
=&  \left| \frac{t^{(-1)^{k+1}(2k-1)+(-1)^{n+1}2n}+t^{(-1)^{n+1}3} }{2i} \right | 
\end{align*}
Note that
$$\left |\frac{t^{a}+t^{b}}{2i} \right |= \left |t^{2n-(a+b)/2}  \cdot \frac{t^{a}+t^{b}}{2i} \right |= \left |\sin\left ( \frac{4n+a-b}{8n} \pi \right ) \right |$$
One obtains for $1\le k \le n-2 $
$$\left | \Lambda \cdot \nu_{n-k} \right| = \left | \sin \left (\frac{4n+(-1)^{n+1}(2n-3)+(-1)^{k+1}(2k-1)}{8n} \cdot \pi  \right ) \right | $$
If $n$ is even,  one obtains the values
\begin{align*}
\sin \left(\frac{\frac{n}{2}+1}{2n} \cdot \pi \right), \sin \left(\frac{\frac{n}{2}}{2n} \cdot \pi \right),\sin \left(\frac{\frac{n}{2}+2}{2n} \cdot \pi \right), \sin \left(\frac{\frac{n}{2}-1}{2n} \cdot \pi \right),\cdots \\
\cdots,  \sin \left(\frac{n-1}{2n} \cdot \pi \right), \sin \left(\frac{2}{2n} \cdot \pi \right)
\end{align*}
If $n=2m+1$ is odd,  one obtains the values 
\begin{align*}
\sin \left(\frac{3m+1}{2n} \cdot \pi \right), \sin \left(\frac{3m}{2n} \cdot \pi \right),\sin \left(\frac{3m+2}{2n} \cdot \pi \right), \sin \left(\frac{3m-1}{2n} \cdot \pi \right),\cdots\\
\cdots ,  \sin \left(\frac{n+1}{2n} \cdot \pi \right), \sin \left(\frac{2n-2}{2n} \cdot \pi \right)
\end{align*}
Hence, in both cases:
$$\left \{ |\Lambda \cdot \nu_{i}| \; \big |  \; 1\le i \le n-2 \right \} = \left \{ \sin \left(\frac{2\pi}{2n} \right ),\cdots, \sin \left(\frac{(n-1)\pi}{2n} \right ) \right \}  $$
Finally
\begin{align*}
\Lambda \cdot \nu_{0} - \Lambda \cdot \eta_1&= \Lambda\cdot (\xi_n+\xi_{n+2}) -\Lambda\cdot (\xi_{n-1}+\xi_n) \\[0.08in]
&= 0
\end{align*}
where we use Corollary \ref{main-corollary} and the fact that the $n-1$ and $n+2$ entries of $\textbf{v}_{1}$ agree. It follows that
$$|\Lambda \cdot \nu_{0}| =\sin \left (\frac{\pi}{2n} \right )$$
This proves the theorem for $\Lambda$ associated to $\textbf{v}_{1}$.  

Suppose now that $\Lambda$ corresponds to $\textbf{v}_{2}$ under the Carter-Weyl correspondence.  To calculate $\Lambda \cdot \eta_{n-k}$ for $1\le k \le n-1$ note that the relevant entries of $\textbf{v}_{1}$ and $\textbf{v}_{2}$ all agree, except the $n$'th entry now is $1/2$ instead of $1$. Therefore, using the earlier calculations for $\textbf{v}_{1}$,  one obtains
\begin{align*}
\Lambda \cdot \eta_{n-k}  &= \frac{t^{(-1)^{k+1}(2k-1)}-t^{(-1)^{n}(2n+1)}+(-1)^{n}t^{(-1)^{n}}(t^{2n}-t^{-2n})/2}{2i}\\[0.17in]
&= \frac{t^{(-1)^{k+1}(2k-1)}-(t^{2n+(-1)^{n}}+t^{-2n+(-1)^{n}})/2}{2i}\\[0.17in]
&= \frac{t^{(-1)^{k+1}(2k-1)}}{2i}
\end{align*}
where the last equation holds since $t^{2n}=-t^{-2n}$. It follows that
$$\left | \Lambda \cdot \eta_{n-k}  \right | = \frac{1}{2}$$
Recall that for $2\le i \le n-1$ one has $\nu_i=\eta_i + \xi_{n+2}$.  Since the $n+2$ entry of $\textbf{v}_{2}$ is $0$,  it follows 
$$|\Lambda \cdot \nu_{i}|=\frac{1}{2}$$
Finally
\begin{align*}
\Lambda \cdot \nu_0 &=\Lambda \cdot (\xi_{n}+ \xi_{n+2})\\[0.08in]
&=\Lambda \cdot \xi_n \\[0.08in]
&= (-1)^{n+1} t^{(-1)^{n}} \cdot \frac{1}{2}
\end{align*}
Hence, after scaling $\Lambda$, all these absolute values equal $1$.
\end{proof}
As an example,  we plot below the Weyl-Orbit values  $m_1\le m_2\le \cdots$ for the Coxeter case $\textrm{D}_{14}$ as well as for $\textrm{D}_{14}(a_6)_{\textbf{I}}$, normalized so that the first non-zero $m_i$ equals $1$.

\hspace{0.2in}

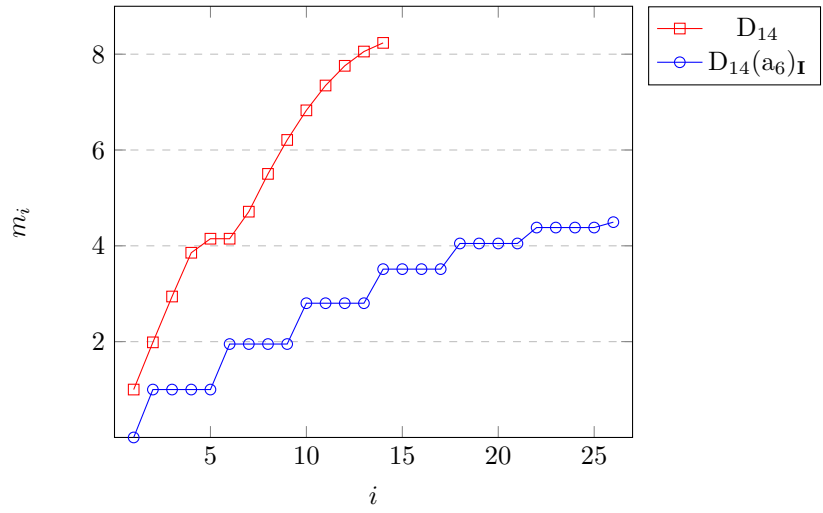
\begin{figure}[h!]
\caption{Normalized Weyl-Orbit values}
$$
\begin{tikzpicture}
\begin{axis}[
    title={},
    xlabel={$i$},
    ylabel={$m_i$},
    xmin=0, xmax=27,
    ymin=0, ymax=9,
    xtick={5,10,15,20,25},
    ytick={2,4,6,8},
    legend pos=outer north east,
    ymajorgrids=true,
    grid style=dashed,
]

\addplot[
	color=red,
    mark=square,
    ]
    coordinates {
    (1,1) 
	(2,1.98542)
	(3,2.94188)
	(4,3.85545)
	(5,4.14811)
	(6,4.14811)	
	(7,4.71280)
	(8,5.50142)
	(9,6.20982)
	(10,6.82766)
	(11,7.34595)
	(12,7.75711)
	(13,8.05516)
	(14,8.23574)

    };
   \addlegendentry{$\textrm{D}_{14}$}  
    
\addplot[
	color=blue,
    mark=o,
    ]
    coordinates {
    (1,0)  
    (2,1)(3,1)(4,1)(5,1) 
    (6,1.94985)(7,1.94985)(8,1.94985)(9,1.94985) 
    (10,2.80194)(11,2.80194)(12,2.80194)(13,2.80194)   
    (14,3.51351)(15,3.51351)(16,3.51351)(17,3.51351)   
    (18,4.04892)(19,4.04892) (20,4.04892)(21,4.04892)   
    (22,4.38129)(23,4.38129)(24,4.38129)(25,4.38129) 
    (26,4.49396)
    };
    \addlegendentry{$\textrm{D}_{14}(\textrm{a}_6)_{\textbf{I}}$}

\end{axis}
\end{tikzpicture}
$$
\end{figure}

\section{Conclusions}

We used a very general correspondence between the linear algebra of Weyl group elements and generalized Cartan matrices to calculate the mass spectrum of Toda-Weyl theories associated to a Weyl group element $\sigma$.  These mass expressions necessitate an understanding of the relation between the Carter roots involved in the definition of $\sigma$, and roots that are orbit representatives of the $\sigma$ action on the set of all roots. The details of this relation depend on the Weyl group element and we carried out this comparison for an interesting infinite family of conjugacy classes in Weyl groups of Lie algebras of type D. The corresponding mass expressions have a trigonometric form closely resembling those for usual affine Toda theories.

\section{Acknowledgements}

It is a pleasure to thank the referee for detailed remarks improving the exposition.


\begin{thebibliography}{99}
 \bibitem[1]{BLM} S.  Berman,  Y. S. Lee,  R. V. Moody: The spectrum of a Coxeter transformation, affine Coxeter transformations, and the defect map, J. Algebra \textbf{121} (1989), 339-357
 \bibitem[2]{BG} D. Borthwick, S. Garibaldi: Did a 1-dimensional magnet detect a 248-dimensional algebra?, Notices of the AMS \textbf{58} (2011), 1055-1066
\bibitem[3]{BOURBAKI} N. Bourbaki: Groupes et Alg\` ebres de Lie, Chapitres 4, 5 et 6. Hermann, Paris,  (1968)
\bibitem[4]{BOU} P. Bouwknegt: Lie algebra automorphisms, the Weyl group, and tables of shift vectors, J. Math. Phys. \textbf{30} (1989), 571-584
\bibitem[5]{BS} L. Brillon, V. Schechtman: Coxeter element and particle masses, Selecta Math.  \textbf{22} (2016),  2591-2609
\bibitem[6]{CAR} R. W. Carter: Conjugacy classes in the Weyl group, Compositio Math. \textbf{25} (1972),  1-59
\bibitem[7]{COL} R. Coldea,  D. A. Tennant,  E. M. Wheeler,  E. Wawrzynska, D. Prabhakaran,  M. Telling,  K. Habicht,  P. Smeibidl,  K. Kiefer: Quantum criticality in an Ising Chain: Experimental evidence for emergent $\textrm{E}_8$ symmetry, Science \textbf{327} (2010),  177-180
\bibitem[8]{COLE} A. J. Coleman: The Betti numbers of the simple Lie groups, Canadian J.  Math.  \textbf{10} (1958),  349-356
\bibitem[9]{COX} H. S. M. Coxeter: The product of the generators of a finite group generated by reflections,  Duke Math.  J.  \textbf{18} (1951),  765-782
\bibitem[10]{DOR} P. Dorey: Root systems and purely elastic S-matrices,  Nucl. Phys. B (1991),  654-676
\bibitem[11]{FRE} M. D. Freeman: On the mass spectrum of affine Toda field theory,  Phys. Lett.  B \textbf{261},  57-61
\bibitem[12]{FLO} A.  Fring,  H.  C.  Liao,  D. I.  Olive: The mass spectrum and coupling in affine Toda theories,  Phys.  Lett.  B \textbf{266} (1991),  82-86
\bibitem[13]{KOS} B. Kostant: The principal three-dimensional subgroup and the Betti numbers of a complex simple Lie group, Amer. J. Math. (1959),  973-1032
\bibitem[14]{LUU} M. T. Luu: The Toda-Weyl mass spectrum,  Nuclear Physics B \textbf{1012} (2025), 116823
\bibitem[15]{REE} M. Reeder: Torsion automorphisms of simple Lie algebras,  L’Enseignement Math. \textbf{56} (2010), 3-47
\bibitem[16]{SPR} T.  A. Springer: Regular elements of finite reflection group, Inventiones Math.  \textbf{25} (1974), 159-198
\bibitem[17]{STE} R. Stekolshchik: Equivalence of Carter diagrams,  Algebra and Discrete Math. \textbf{23} (2017), 138-179









\end{thebibliography}
\end{document}